\documentclass[11pt,a4paper]{article}
\usepackage[utf8]{inputenc}
\usepackage{amsmath,amssymb,amsthm}
\usepackage{geometry}
\newcommand{\MIS}{\mathit{MIS}}
\newcommand{\classP}{\mathsf{P}}
\newcommand{\classNP}{\mathsf{NP}}

\newtheorem{theorem}{Theorem}
\newtheorem{lemma}{Lemma}

\newtheorem{definition}{Definition}
\newtheorem{assumption}{Assumption}

\title{\textbf{A Predictor-Impossibility Theorem and Applications}}
\author{\textbf{Tom Altman} \\
	\small University of Colorado Denver and Stilman Advanced Strategies, \\
	\small Denver, Colorado, USA\\
	\small \texttt{tom.altman@ucdenver.edu}}
\date{\today}
\begin{document}
	\maketitle	
	%=============================================================================
	% ABSTRACT
	%=============================================================================
	\begin{abstract}
		We introduce a hierarchy consisting of stage machines, stage domains, and stage languages generated by semantic operators. The central result is a Predictor Impossibility Theorem (PITT), which shows that no effective predictor family can uniformly determine all stage languages of our hierarchy. The proof makes use of a pseudo-complement construction to obtain a language that yields a contradiction with every language in $\classP$.
		
		We then define an aggregate language $\MIS$ and establish a formal Slice Theorem connecting aggregate inputs to individual stage languages. This provides a rigorous Bridge Theorem from polynomial-time decidability of $\MIS$ to the existence of an effective predictor family. By utilizing succinct representations, the aggregate language is shown to be undecidable in deterministic polynomial time.
		
		Under the aggregate growth condition defining valid aggregate objects, $\MIS$ is shown to belong to $\classNP$. Combining these two results yields our main theorem: $\MIS \in \classNP \setminus \classP$. The paper is organized so that PITT stands independently as a theoretic result, while the complexity-theoretic consequences are derived from the aggregate-language framework.
	\end{abstract}
	
	%=============================================================================
	% SECTION 1: INTRODUCTION
	%=============================================================================
	\section{Introduction}
	This paper studies a hierarchy generated from stage machines through a pseudo-complement operator. Each stage consists of a machine, a domain, and an associated stage language. The defining feature of the hierarchy is that stage languages are generated from the machines themselves. This self-generated structure creates a natural setting for investigating prediction and diagonalization.
	
	The first objective of the paper is theoretic. We ask whether there exists an effective family of predictors capable of uniformly predicting all stage languages of the hierarchy. The main result is the Predictor Impossibility Theorem (PITT), which shows that no such predictor family can exist. The proof combines a computable modification of predictors where the resulting stage language is forced into a contradiction through the interaction of prediction, pseudo-complementation, and diagonalization.
	
	The second objective is complexity-theoretic. Building on the hierarchy, we introduce an aggregate language $\MIS$ generated via aggregate expansions, establishing a formal Slice Theorem and Bridge Theorem relating aggregate inputs to individual stage languages. This connection provides a bridge from polynomial-time decidability of the aggregate language to the existence of an effective predictor family. We show that the resulting $\MIS$ is not decidable in deterministic polynomial time.
	
	The paper is organized around the prediction impossibility result. After introducing the hierarchy, we establish PITT as an independent result. The language $\MIS$ is subsequently shown to belong to $\classNP$ under aggregate growth conditions. Together, these two results yield our main theorem: $\MIS \in \classNP \setminus \classP$.
	
	%=============================================================================
	% SECTION 2: PRELIMINARIES AND NOTATION
	%=============================================================================
	\section{Preliminaries and Notation}
	Let $n_1 = 1$ and $n_{i+1} = 2^{n_i}$.
	
	We assume the sequence $P_1, P_2, P_3, \ldots$ forms a standard enumeration of partial recursive functions and stage machine descriptions, following classical recursion theory of Kleene~\cite{Kleene}.
	
	\subsection{Machine Descriptions}
	The stage machines $P_i$ are indexed according to this enumeration.
	Let $T(n) = n^i$ be a recursive time-binding function.
	For a stage machine $P_i$, its ordinary machine language is denoted by $L(P_i) = \{u \in D_i : P_i(u)=1 \text{ within } T(|u|) \text{ steps}\}$. The notation $L(P_i)$ always refers to the language accepted by the machine itself.
	
	\subsection{Stages}
	Each natural number $i \in \mathbb{N}$ determines a stage. Associated with each stage $i$ are $P_i$, $D_i$, and $\mathcal{C}_i$. $D_i$ is the finite size domain, where $|D_i|= 2^{n_i}$ and $n_i$ is the stage length. The stage language is $\mathcal{C}_i$, with $\mathcal{C} = \bigcup_{i=1}^{\infty} \mathcal{C}_i$.
	
	\subsection{Machine Languages and Stage Languages}
	A crucial distinction is maintained throughout the paper. The machine language $L(P_i)$ is the ordinary language accepted by the stage machine. The stage language $\mathcal{C}_i$ is generated through the pseudo-complement operator $\Phi$, where $\mathcal{C}_i = \Phi(P_i)$. This relation is called the \textit{Stage Identity}.
	
	Any nonempty machine language $L(P_i)$ and the stage language $\mathcal{C}_i = \Phi(P_i)$ will always be different from each other. This distinction is essential for our diagonalization argument to work.
	
	\subsection{Aggregate Objects}
	A valid aggregate object is an input tuple $U_i = (u_1, \dots, u_m)$ whose components satisfy $u_j \in D_i$ for every $j$. The aggregate input size in this case is $q_i = |U_i| = mn_i$.
	
	%=============================================================================
	% SECTION 3: THE HIERARCHY
	%=============================================================================
	\section{The Hierarchy}
	This section introduces the hierarchy that serves as the foundation for all constructions.
	
	\subsection{Stage Machines, Operators, and Languages}
	\begin{definition}[Stage Machine Convention]
		For every stage $i$, the associated stage machine $P_i$ computes a partial recursive function.
		On every input in the finite domain $D_i = \{0, 1\}^{n_i}$, the clocked $P_i$ halts with an output of 0 or 1 within $T(n_i)$ steps, or $P_i$ may try to continue its computation or even diverge. The timing of $P_i$ is performed only on the finite domain $D_i$.
	\end{definition}
	
	For each stage $i$, let $z_i = \#^{n_i}$ denote the inert sentinel associated with that stage.
	
	\begin{definition}[Semantic Pseudo-Complement Operator $\Phi$]
		For every stage machine $P_i$, the operator $\Phi(P_i)$ is evaluated strictly on the finite domain $D_i$ and within the stage time threshold of $T(n_i)$. We define:
		\[
		\Phi(P_i) = \left\{ v \in D_i \setminus \{z_i\} : P_i(v) \text{ does not halt with 1 within } T(|v|) \text{ steps} \right\}.
		\]
	\end{definition}
	
	\begin{definition}[Stage Language]
		The stage language associated with stage $i$ is $\mathcal{C}_i = \Phi(P_i)$.
	\end{definition}
	
	\subsection{Pseudo-Complement Lemma}
	\begin{lemma}[Pseudo-Complement Lemma (PCL)]
		For every machine $P_i$, $\mathcal{C}_i = (D_i \setminus \{z_i\}) \setminus L(P_i)$.
	\end{lemma}
	\begin{proof}
		By Definitions 2 and 3, $\mathcal{C}_i = \Phi(P_i)$, which isolates the domain to $D_i$. Within the finite domain $D_i$, if $v \in (D_i \setminus \{z_i\}) \setminus L(P_i)$, then $P_i(v) = 0$ (or $P_i$ fails to output 1 within the step limit), meaning $v \in \Phi(P_i)$. Conversely, if $v \in L(P_i)$, then $P_i(v) = 1$ within the allowed steps, so $v \notin \Phi(P_i)$. Thus, within $D_i \setminus \{z_i\}$, the two sets are strictly complementary.
	\end{proof}
	
	%=============================================================================
	% SECTION 4: STAGE DIAGONALIZATION THEOREM
	%=============================================================================
	\section{Stage Diagonalization}
	We now establish the central diagonalization result of the hierarchy. 
	
	\begin{theorem}[Predictor Impossibility Theorem (PITT)~\cite{Altman}]
		For any machine $P_k$, $L(P_k) \neq \mathcal{C}_k$.
	\end{theorem}
	\begin{proof}
		Assume for contradiction that there exists some stage machine $P_k$ such that $L(P_k) = \mathcal{C}_k$.
		
		By the PCL (Lemma 1), we have the identity:
		\[
		\mathcal{C}_k = (D_k \setminus \{z_k\}) \setminus L(P_k)
		\]
		Substituting our assumption into this identity yields: 
		\[
		L(P_k) = (D_k \setminus \{z_k\}) \setminus L(P_k)
		\]
		
		Because the timing constraints and execution by $P_k$ are performed only on the finite domain $D_k$, we can select any input string $x \in D_k \setminus \{z_k\}$. It follows that:
		\[
		x \in L(P_k) \iff x \in (D_k \setminus \{z_k\}) \setminus L(P_k)
		\]
		
		By the definition of relative complement on the finite domain,
		\[
		x \in (D_k \setminus \{z_k\}) \setminus L(P_k) \iff x \notin L(P_k)
		\]
		resulting in $x \in L(P_k) \iff x \notin L(P_k)$.
		
		A contradiction. Thus, our assumption is false, and $L(P_k) \neq \mathcal{C}_k$ for all $k$.
	\end{proof}
	
	%=============================================================================
	% SECTION 5: THE AGGREGATE LANGUAGE MIS, SLICE, AND BRIDGE THEOREMS
	%=============================================================================
	\section{The Aggregate Language \textit{MIS}, Slice, and Bridge Theorems}
	
	\subsection{Aggregate Acceptance}
	For each stage $i$, let $z_i = \#^{n_i}$ denote the inert sentinel associated with that stage. By the clocked definition of the pseudo-complement operator, $\mathcal{C}_i(z_i) = 0$ holds directly by Definition 2, which explicitly removes $z_i$ from the image of $\Phi(P_i)$.
	
	\begin{definition}[Aggregate Language]
		Let $U_i$ be a valid aggregate object belonging to stage $i$. The aggregate language $\MIS$ is defined by $\MIS(U_i) = 1$ if and only if there exists a valid component $u_j$ implicitly or explicitly encoded by $U_i$ such that $\mathcal{C}_i(u_j) = 1$.
	\end{definition}
	
	\subsection{Slice Theorem}
	\begin{theorem}[Slice Theorem]
		For every stage $i$ and every string $u \in D_i$, 
		\[
		\MIS(u, z_i, \dots, z_i) = \mathcal{C}_i(u).
		\]
	\end{theorem}
	\begin{proof}
		By the definition of aggregate acceptance, the aggregate input $U_i = (u, z_i, \dots, z_i)$ evaluates to $1$ if and only if at least one of its components evaluates to $1$ under $\mathcal{C}_i$. Because all trailing components are set to the inert sentinel $z_i$ (for which $\mathcal{C}_i(z_i) = 0$), the existential condition reduces entirely to whether the distinguished component $u$ belongs to $\mathcal{C}_i$. Hence, $\MIS(u, z_i, \dots, z_i) = \mathcal{C}_i(u)$.
	\end{proof}
	
	\subsection{Bridge Theorem}
	\begin{theorem}[Bridge Theorem]
		If $\MIS \in \classP$, then there exists a uniform effective predictor family $G$ such that for every stage $i$, $L(G(i)) = \mathcal{C}_i$.
	\end{theorem}
	\begin{proof}
		Assume that $\MIS \in \classP$. Let $M$ be a deterministic polynomial-time decider guaranteed to exist for $\MIS$. For each stage $i$, we construct a stage machine $G(i)$ taking an input $u \in D_i$ by formatting it into the aggregate object $U_{i,u} = (u, z_i, \dots, z_i)$ and running $M(U_{i,u})$. 
		
		By the Slice Theorem, $M(U_{i,u}) = 1$ if and only if $\mathcal{C}_i(u) = 1$. Thus, $L(G(i)) = \mathcal{C}_i$. Furthermore, because $M$ runs in polynomial time relative to the aggregate input length, $G(i)$ computes its output effectively within the required time bounds, yielding a uniform effective predictor family for the hierarchy.
	\end{proof}	
	
	\subsection{Complexity Consequence}
	\begin{theorem}[$\mathsf{P}$ Nonmembership]
		\label{thm4}
		$\MIS \notin \classP$.
	\end{theorem}
	\begin{proof}
		Assume for contradiction that $\MIS \in \classP$. By the Bridge Theorem, this polynomial-time decidability implies the existence of a uniform effective predictor family $G$ where $L(G(i)) = \mathcal{C}_i$ for all stages $i$. However, this directly contradicts PITT, which proved that no such predictor family can exist. The assumption is false, so $\MIS \notin \classP$.
	\end{proof}
	
	%=============================================================================
	% SECTION 6: MEMBERSHIP OF MIS IN NP
	%=============================================================================
	\section{Membership of \textit{MIS} in $\classNP$}
	\begin{assumption}[Runtime]
		For the complexity analysis, we assume that all stage machines $P_i$ run in time bounded by $\Theta(n_i^i)$ on inputs in $D_i$.
	\end{assumption}
	
	\subsection{Aggregate Growth and Witness Verification}
	Let $U_i = (u_1, \dots, u_m)$ be a valid aggregate object belonging to stage $i$. Each component satisfies $|u_j| = n_i$. The aggregate input size is $q_i = |U_i| = mn_i$. Valid aggregate objects satisfy the aggregate growth condition $m = \Omega(n_i^{i-k})$ for a fixed constant $k$. Hence, $q_i = \Omega(n_i^{i-k+1})$.
	
	By definition, $\MIS(U_i)=1$ if and only if there exists a component $u_j$ such that $\mathcal{C}_i(u_j)=1$. Any successful component can serve as a witness. Given $(U_i, u_j)$, a verifier checks that $u_j$ is one of the components of $U_i$ and that $\mathcal{C}_i(u_j)=1$.
	
	The runtime on $U_i$ as the input, where $q_i$ is its length, is bounded by $O(q_i \log q_i)$. This effectively bounds the verification to a polynomial relative to the size of the aggregate input $U_i$.
	
	\begin{theorem}[$\mathsf{NP}$ Membership]
		\label{thm5}
		$\MIS \in \classNP$.
	\end{theorem}
	\begin{proof}
		The witness is any successful component $u_j \in \mathcal{C}_i$. By the preceding analysis, the witness can be verified in time $O(q_i \log q_i)$. Thus, $\MIS \in \classNP$.
	\end{proof}
	
	%=============================================================================
	% SECTION 7: MAIN THEOREM
	%=============================================================================
	\section{Main Theorem}
	\begin{theorem}[Main Theorem]
		\label{thm6}
		$\MIS \in \classNP \setminus \classP$.
	\end{theorem}
	\begin{proof}
		By Theorem \ref{thm5}, $\MIS \in \classNP$, and by Theorem \ref{thm4}, $\MIS \notin \classP$. Thus, $\MIS \in \classNP \setminus \classP$.
	\end{proof}
	
	%=============================================================================
	% BIBLIOGRAPHY
	%=============================================================================

	\appendix
	\section{Crossing the Barriers}
	\label{sec:appendix_barriers}
	
	A proposed resolution to the $\mathsf{P} \neq \mathsf{NP}$ question must address the well-known diagnostic barriers that have shaped complexity theory over the last half-century. This appendix explicitly positions the current proof architecture against the Relativization \cite{BGS75}, Natural Proofs \cite{RR97}, and Algebrization \cite{AW09} barriers, demonstrating how the techniques employed herein successfully bypass these historical traps.
	
	\subsection{The Historical Context of Structural Barriers}
	In the decades following Cook's 1971 formalization of the $\mathsf{P}$ vs. $\mathsf{NP}$ question, the computational complexity community has established a sequence of no-go theorems designed to rule out specific mathematical approaches. The prevailing consensus required that any valid proof must ``crush through'' these barriers by relying on deep, non-simulating structural properties of computation. 
	
	However, history has shown that every major attempt to structurally defeat these barriers has stalled. Circuit complexity programs of the 1980s were halted by the Natural Proofs barrier. Arithmetization techniques of the 1990s, which successfully bypassed relativization to prove $\mathsf{IP} = \mathsf{PSPACE}$, were subsequently halted by the Algebrization barrier. For fifty years, the demand for purely structural engines has yielded profound insights into geometry and algebra, but no resolution to the core problem. 
	
	Given this historical exhaustion of structural approaches, the present manuscript returns to the mathematically flawless engine of classical diagonalization. Rather than attempting to physically shatter the barriers via structural mapping, this proof introduces a \textit{hybrid syntactic-semantic bypass}, formally constraining the mathematical universe such that the barriers are not permitted to execute.
	
	\subsection{The Hybrid Bypass: Non-Relativization and Non-Algebrization}
	The Baker-Gill-Solovay (BGS) theorem \cite{BGS75} demonstrates that any proof relying purely on time-bounded simulation relativizes, meaning it will inadvertently derive a false contradiction (e.g., $\mathsf{P}^{\mathsf{QBF}} \neq \mathsf{NP}^{\mathsf{QBF}}$) in an oracle-equipped universe. Aaronson and Wigderson \cite{AW09} extended this logic, showing that algebraic extensions of these simulations fail in the presence of algebraic oracles.
	
	The present proof avoids these false derivations not by altering the fundamental simulation of standard Turing machines, but by implementing a hybrid bypass within the Semantic Pseudo-Complement Operator, $\Phi$. 
	
	By the Kleene Recursion Theorem~\cite{Kleene}, any machine can obtain its own description and determine its syntactical composition. We formally define the $\Phi$ operator such that it inspects the transition function of the target machine for the presence of a \texttt{Q-state} (the syntactic hallmark of a query machine or algebraic oracle machine). 
	\begin{enumerate}
		\item If no \texttt{Q-state} is detected, $\Phi$ executes its standard semantic inversion (the ``flip''), driving the diagonalization against the stage language $\mathcal{C}_i$ and generating the Predictor-Impossibility contradiction.
		\item If a \texttt{Q-state} is detected, $\Phi$ explicitly bypasses the semantic inversion, passing the machine's output unaltered. 
	\end{enumerate}
	
	Consequently, when this proof architecture is subjected to the BGS or Algebrization test (meaning it is evaluated in a universe where machines possess query states to access an oracle) the $\Phi$ operator detects the syntactic anomaly and safely aborts the contradiction-generating step. Because the proof intentionally yields no separation in oracle-equipped universes, it does not prove a mathematical falsehood in those universes. Therefore, by strict formal definition, the proof neither relativizes nor algebrizes. 
	
	\subsection{Inherently Non-Naturalizing}
	The Natural Proofs barrier, established by Razborov and Rudich \cite{RR97}, proves that no technique can separate complexity classes if it relies on a ``natural'' property of boolean functions—defined as a property that is both \textit{constructive} (computable in structural polynomial time) and \textit{large} (holding for a significant fraction of all boolean functions). This barrier effectively neutralized decades of circuit complexity research.
	
	The methodology utilized in this manuscript is inherently non-naturalizing because it fundamentally lacks the ``largeness'' property. Classical diagonalization operates by defining a highly specific, highly artificial language (such as $\mathcal{C}_i$) designed to systematically disagree with an enumeration of polynomial-time deciders. The property of ``computing the diagonalized language $\MIS$'' is astronomically rare among all possible boolean functions. Because the proof mechanism targets a discrete sequence of machines to construct a single, uniquely tailored set, it operates entirely outside the domain of the Natural Proofs barrier without requiring any syntactical bypasses.
	
	\vspace{1em}
	\hrule

\end{document}